\documentclass[11.5pt]{amsart} 
\usepackage{euscript,amsmath,amssymb,amsfonts,graphicx,bm,amsthm,mathtools}
\usepackage{cite}
\usepackage{paralist}
\usepackage{graphics} 
\usepackage[caption=false]{subfig}
\usepackage{latexsym,enumerate} 
\usepackage{tikz}
\usepackage[colorlinks=false]{hyperref}
\hypersetup{urlcolor=blue, citecolor=red}

\usepackage{pgfplots}
\usepackage{soul,xcolor}

\newtheorem{theorem}{Theorem}

\newtheorem{proposition}[theorem]{Proposition}
\newtheorem{corollary}[theorem]{Corollary}
\theoremstyle{plain}
\theoremstyle{remark}

\theoremstyle{definition}

\usepackage{abstract}

\usepackage{adjustbox}

\begin{document}
\title{Hitting probabilities for fast stochastic search}
\author{Samantha Linn}
\thanks{SL (corresponding author): Department of Mathematics, University of Utah, Salt Lake City, UT 84112 USA. (\texttt{slinn.math@gmail.com}). SL was supported by the National Science Foundation (Grant No.\, 2139322)}
\author{Sean D. Lawley}
\thanks{SDL: Department of Mathematics, University of Utah, Salt Lake City, UT 84112 USA. (\texttt{lawley@math.utah.edu}). SDL was supported by the National Science Foundation (Grant Nos. DMS-2325258 and DMS-1944574).}
\maketitle

\begin{abstract}
Many physical phenomena are modeled as stochastic searchers looking for targets. In these models, the probability that a searcher finds a particular target, its so-called hitting probability, is often of considerable interest. In this work we determine hitting probabilities for stochastic search processes conditioned on being faster than a random short time. Such times have been used to model stochastic resetting or stochastic inactivation. These results apply to any search process, diffusive or otherwise, whose unconditional short-time behavior can be adequately approximated, which we characterize for broad classes of stochastic search. We illustrate these results in several examples and show that the conditional hitting probabilities depend predominantly on the relative geodesic lengths between the initial position of the searcher and the targets. Finally, we apply these results to a canonical evidence accumulation model for decision making.
\end{abstract}

\section{Introduction}
Various physical phenomena are often modeled as stochastic `searchers' looking for `targets'. When there is more than one target, the probability that a searcher finds a particular target, its so-called hitting probability, is typically of interest. In recent years there has been significant progress in characterizing hitting probabilities associated with stochastic search processes \cite{kesten1987, dalang2009,PhysRevE.102.022115,PhysRevLett.126.100602,Klinger2022}. Here, we study the hitting probabilities associated with stochastic search conditioned on being faster than a random short time. Given the generality of its construction, the results herein are useful in numerous areas of applied mathematics and statistical physics research including search under stochastic resetting or stochastic inactivation.

Stochastic resetting describes the random repositioning of a searcher according to a given distribution. The utility of stochastic resetting often lies in reducing the mean first passage time (FPT) of search processes or otherwise optimizing search processes \cite{Bressloff_2020_gatedtargets,gonzalez2021, schumm2021,Ray2022,tucci2022, Liu2023,guo2023,zbik2023levy}. Consider, for example, the work of Mercado-V\'asquez and Boyer \cite{Mercado_2018}, which models predator dynamics as L\'evy flights that stochastically reset to regions of prey. Instead of assessing strategy optimality in terms of search times, the authors use predator population as a proxy. Their analysis reveals an interesting relationship between mobility, resetting, birth, death, and total population.

Stochastic inactivation, which is used to assign a lifespan to a searcher according to a given distribution, can also be used to tune, or sharpen, search process statistics \cite{yuste2013,meerson2015, grebenkov2017,radice2023, boyer2024optimizing}. In the work of Ma et al \cite{ma2020} on the role of inactivation in intracellular signalling, inactivation can model degradation, (de)phosphorylation, or other mechanisms that immobilize signal propagation. The authors use soft X-ray tomography images of human B cells to study the stochastic search process of a stochastically inactivating intracellular signal from cell membrane to nucleus. By numerically solving the Fokker-Planck equation associated with the propagation of the signal, they determine the full distribution of the arrival time of the signal to the nucleus and quantify how stochastic inactivation can compensate for the delay in signal arrival times due to organelle barriers.

Much of the current literature involving stochastic resetting or inactivation, like the examples described above, considers only exponentially distributed random times. While often suitable, especially in memoryless systems, by no means does this choice exhaust all such conditional search processes \cite{pal2017,evans2020,FPTuFSR}. For example, stochastic inactivation is a one-time event and thus agnostic to system memory. Moreover, it is typical that memory plays a critical role in population-level search processes. For these reasons, and also to emphasize the ease with which our methods can be adapted, we consider a variety of distributions for this time.

In particular, we consider an evidence accumulation model for fast decision making before a short, deterministic time. By equating evidence with the log-likelihood ratio of possible hypotheses, or choices, experimental data suggests that models of evidence accumulation accurately capture decision-making and, moreover, that this process evolves according to a biased Brownian motion \cite{swets1961,ratcliff1978,newsome1989, banburismus2002,chittka2003,uchida2003,gold2007}. Such models have thus been  widely used to study how humans and other animals make choices \cite{bogacz2006,mulder2012, mann2018,karamched2020prl,karamched2020,tump2022,reina2023,stickler2023}.

In the fast search limit, only an approximation of the short-time behavior of the unconditional search process is needed to compute the conditional hitting probabilities. To be precise, consider an unconditional search process and let $\{X(t)\}_{t\geq 0}$ denote the path of the stochastic searcher in a domain $\mathcal{D}\in\mathbb{R}^d$ with $K\geq 2$ targets denoted by $V_0$, $V_1$, $\dots$, $V_{K-1}\in\mathcal{D}$. Then the first hitting time to a target is
\begin{align*}
    \tau := \text{inf}\{t>0 : X(t) \in \cup_{k=0}^{K-1}V_k\in\mathcal{D}\}.
\end{align*}
Let $\kappa \in \{0,1,\dots,K-1\}$ denote the index of the target hit by the searcher. Then the hitting probability of target $V_k$ is
\begin{align*}
    \mathbb{P}(\kappa = k)\quad \text{for }k\in\{0,1,\dots,K-1\}.
\end{align*}
Analytically computing the hitting probabilities for a purely diffusive search requires solving an elliptic partial differential equation with mixed boundary conditions and, depending on the domain, can be quite involved \cite{oksendal2003}. A similarly involved procedure applies to computing statistics of first hitting times \cite{gardiner2009}.

While these results are useful in numerous settings, it has been emphasized that other timescales are often more relevant. For instance, the term `redundancy principle' has recently been coined to convey the utility of seemingly wasteful copies of biological entities in accelerating time-sensitive processes \cite{schuss2019,lawley2020uni}. In this case, the relevant time scale is that of the fastest searcher, called the extreme FPT, and the corresponding hitting probabilities are called extreme hitting probabilities \cite{Linn_2022}.

In another instance, that of interest in this work, there may be a time limit before which an event must occur. Such a time limit may occur naturally or can serve as an external force aiming to filter out certain perhaps undesirable behaviors and provide some amount of reliability, or predictability, to the system. For example, it was recently shown that early arrivers to targets encode information about search initial condition \cite{lindsay2023}. Throughout, we refer to this time limit in relation to stochastic resetting but we emphasize that these results hold much more generally. Moreover, we remark that while first hitting \emph{times} are not agnostic to whether the time constraint indicates inactivation or resetting, hitting \emph{probabilities} are.

\begin{figure}[h!]
  \centering
  \includegraphics[width=8cm]{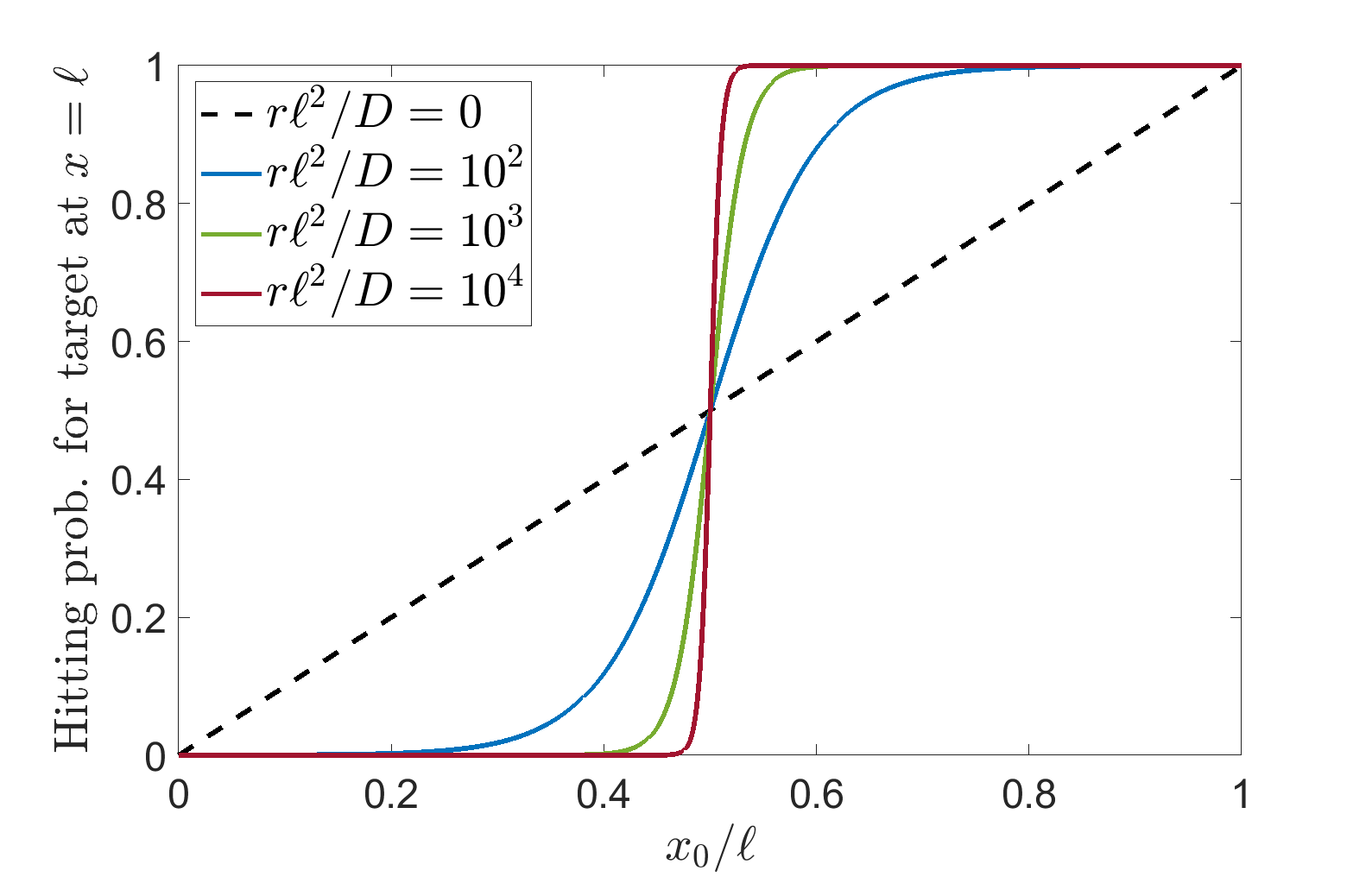}
 \caption{Hitting probabilities for diffusion with unit diffusivity, initial position $x_0\in(0,\ell)$. Solid colored curves illustrate case of exponentially distributed resetting times $\sigma>0$ with rate $r>0$; dashed black curve denotes no conditioning.}
 \label{fig0}
\end{figure}

As a motivating example, consider a driftless diffusive search process with diffusivity $D>0$ in one dimension. Assume the searcher starts at $x_0\in(0,\ell/2)$ where $V_0=(-\infty,0]$ and $V_1=[\ell,\infty)$ are targets. Let $L_0>0$ and $L_1>0$ be the distances to $V_0$ and $V_1$ from $x_0$, respectively. In the case that there is no condition on the search time, one can show \cite{redner2001} that the probability of the searcher reaching $V_1$ before $V_0$ is \[\mathbb{P}(\kappa = 1) = \frac{x_0}{\ell} = 1 - \mathbb{P}(\kappa = 0).\]
Now suppose the searcher stochastically resets to its initial position $x_0$ at random independent and identically-distributed times governed by $\sigma>0$. If $\sigma>0$ is exponentially distributed with rate $r>0$, we show in section \ref{pd} that the hitting probability to $V_1$ decays to zero like
\begin{align} \label{hence}
    \mathbb{P}(\kappa=1 | \tau<\sigma) \sim \text{exp}(-\sqrt{(L_1-L_0)^2r/D})
    \quad\text{as }r\to\infty.
\end{align}
In words, Equation \eqref{hence} says that the hitting probability to the further target decays to zero exponentially fast with respect to the square root of the resetting rate. This decay is especially fast when the relative distances to the nearest target and furthest targets is large. Hence, while previous studies often focus on using resetting to accelerate search times, Equation \eqref{hence} shows how resetting can be used to ensure that a particular target is found. We illustrate this behavior in Figure \ref{fig0}. The remainder of this paper concerns this problem in more general settings with various stochastic search processes, various spatial domains of different dimensions, and various resetting distributions.

Throughout, we write the hitting probability of target $V_k$ associated with a resetting searcher by
\begin{align} \label{prob}
    \mathbb{P}(\kappa = k | \tau<\sigma)\quad \text{for }k\in\{0,1,\dots,K-1\}.
\end{align} 
Let $k=0$ be the index of the target nearest the initial position of the searcher. In Equation \eqref{hence} we show how a hitting probability depends on the geodesic lengths to the nearest target and the target of interest from the searcher initial position. We verify in this work that this exponential decay holds in much more general scenarios of diffusive search and, more broadly, that
\begin{align} \label{lim1}
    \mathbb{P}(\kappa = 0 | \tau < \sigma) \to 1\quad\text{in the frequent resetting limit}.
\end{align}
That is, a frequently resetting searcher will necessarily find the nearest target. Though \eqref{lim1} may seem intuitive, we derive quantitative estimates of the convergence in \eqref{lim1} for a wide variety of search scenarios. Moreover, we show that \eqref{lim1} does not hold for every search process of interest.

The rest of the paper is organized as follows. We detail the results section \ref{section2}. In section \ref{section3}, we apply these results to several examples and compare to numerical solutions. In section \ref{section3b}, we apply our results to an evidence accumulation model for decision making. We conclude by discussing our results in the context of recent related work. Proofs and numerical details are collected in the Appendix.

\section{Conditional hitting probability asymptotics} \label{section2}
In this section, we present results on hitting probabilities for searchers conditioned on being faster than a random short time. This section and the results herein make no reference to the underlying search process, instead assuming properties of the short-time behavior of the unconditional hitting probabilities. Mild conditions are placed on the resetting distribution. In section \ref{section3} we apply our results to diffusive search and other stochastic search processes.

\subsection{Probabilistic setup and integral representation} \label{prelim}
Let $\tau>0$ be a strictly positive random variable and let $\kappa$ be a random variable taking values from the set $\{0,1,\dots,K-1\}$. In the applications of this text, $\tau>0$ denotes the hitting time of an unconditional search process to a target and $\kappa$ indexes which of the $K\geq 2$ targets is hit. For each index $k\in \{0,1,\dots,K-1\}$, let $\tau^{(k)}$ be the hitting time to target $V_k$,
\begin{align*}
    \tau^{(k)} = \begin{cases}
        \tau & \kappa = k,\\
        +\infty & \kappa \neq k,
    \end{cases}
\end{align*}
and define
\begin{align*}
    F_{\tau}^{(k)}(t) :=\mathbb{P}(\tau^{(k)} \leq t) = \mathbb{P}(\tau\leq t \cap \kappa = k),\quad t\in \mathbb{R}.
\end{align*}
Let $F_{\tau}(t)$ denote the cumulative distribution function of $\tau>0$,
\begin{align*}
    F_{\tau}(t) := \mathbb{P}(\tau\leq t) = \sum_{k=0}^{K-1} F_{\tau}^{(k)}(t),\quad t\in\mathbb{R}.
\end{align*}

Further, let $\sigma>0$ be a strictly positive random variable defined by 
\begin{align*}
    \sigma := Y/r,
\end{align*}
where we call the parameter $r>0$ the resetting rate and $Y>0$ is a strictly positive random variable with unit mean and a finite moment generating function in a neighborhood of the origin. That is, there exists $\eta>0$ such that
\begin{align} \label{ezy}
\begin{split}
    \mathbb{E}[Y] &= 1, \\
    \mathbb{E}[e^{zY}] &< \infty\quad \text{for all }z\in[-\eta,\eta]. 
\end{split}
\end{align}
Applying the second condition in \eqref{ezy} and Chebyshev's inequality (see, for example, Theorem 1.6.4 in \cite{durrett2019}), we obtain that $S_Y(y) := \mathbb{P}(Y>y)$ decays at least exponentially fast,
\begin{align*}
    S_Y(y) \leq Ce^{-\eta y}\quad\text{for all } y\in\mathbb{R},
\end{align*}
where $\eta>0$ is as in \eqref{ezy} and $C = \mathbb{E}[e^{\eta Y}] < \infty$. Since $\tau$ and $\sigma$ are independent random variables, we can write the probability of the search process ending prior to a resetting event as a Riemann-Stieltjes integral,
\begin{align} \label{pdef}
    p = p(r) := \mathbb{P}(\tau<\sigma) = \mathbb{E}[S_{\sigma}(\tau)] = \int_0^{\infty} S_{\sigma}(t)\,\textup{d}F_{\tau}(t),
\end{align}
where $S_{\sigma}(t) = \mathbb{P}(\sigma>t)$ is the survival probability of $\sigma>0$. Similarly, since $\tau^{(k)}$ and $\sigma$ are independent random variables, the probability of the search process hitting target $\kappa=k$ prior to a resetting event is
\begin{align} \label{pkdef}
    p_k = p_k(r) := \mathbb{P}(\kappa=k \cap \tau<\sigma) = \mathbb{E}[S_{\sigma}(\tau^{(k)})] = \int_0^{\infty} S_{\sigma}(t)\,\textup{d}F_{\tau}^{(k)}(t).
\end{align}
To avoid trivial cases, we assume $p>0$ and $p_k>0$ for all $r>0$. Hence, by definition of conditional probability,
\begin{align} \label{t1}
    \mathbb{P}(\kappa=k | \tau < \sigma) \sim \frac{\tilde{p_k}}{\tilde{p}}\quad\text{as }r\to\infty
\end{align}
where $\tilde{p}=\tilde{p}(r)$ is any function of $r>0$ such that $\tilde{p}\sim p$ as $r\to\infty$ and likewise of $\tilde{p_k}$. Throughout this work, $f\sim g$ means $f/g \to 1$.
The results in subsection \ref{22} allow us to determine explicit expressions for \eqref{t1} in a wide variety of stochastic search processes.

\subsection{Exact hitting probability asymptotics} \label{22}
Applying \eqref{t1} to a specific search process requires specifying a resetting distribution. We consider four such distributions: exponential, gamma, uniform, and sharp (or deterministic). The following proposition, previously stated and proved in \cite{Linn_2022}, states the asymptotic behavior of an integral representative of those in \eqref{pdef} and \eqref{pkdef} in typical scenarios of diffusive search under exponential and gamma distributed resetting.

\begin{proposition} \label{prop3}
    Assume $C>0$, $\delta>0$ and $b\in\mathbb{R}$. Then
    \begin{align} \label{prop3eq}
        \int_0^{\delta} e^{-rt} t^b e^{-C/t} \,\textup{d}t \sim \sqrt{\pi C^{\frac{2b+1}{2}}} r^{\frac{-2b-3}{4}} e^{-\sqrt{4Cr}}  \quad \text{as }r\to\infty.
    \end{align}
\end{proposition}

Theorem \ref{thm3} below uses Proposition \ref{prop3} to compute the asymptotic behavior of \eqref{prob} in the large $r$ limit assuming information about the short-time behavior of $F_{\tau}$ and $F_{\tau}^{(k)}$ on a linear scale in the case that resetting is gamma distributed. Its corollary (Corollary \ref{cor4}) concerns the specific case that resetting is exponentially distributed.
\begin{theorem} \label{thm3}
    Under the assumptions of section \ref{prelim}, assume further that $\sigma>0$ is gamma distributed with shape $\eta>0$ and rate $r>0$ and that for some $k\in\{1,\dots,K-1\}$,
    \begin{align}
        F_{\tau}(t) &\sim At^me^{-C/t}\quad \text{as }t\to 0^+,\label{Ftau1}\\
        F_{\tau}^{(k)}(t) &\sim Bt^ne^{-C_k/t}\quad \text{as }t\to 0^+,\label{Ftauk1}
    \end{align}
    where $A>0$, $B>0$, $C_k>C>0$, and $m,\,n\in\mathbb{R}$. Then
    \begin{align} \label{gam_result}
        \mathbb{P}(\kappa=k | \tau<\sigma) \sim \xi r^{(m-n)/2} e^{-2(\sqrt{C_k}-\sqrt{C})\sqrt{r}} \quad\text{as }r\to\infty,
    \end{align}
    where
    \begin{align*}
        \xi := \frac{B}{A} \frac{C_k^{(2n+2\eta-1)/4}}{C^{(2m+2\eta-1)/4}} > 0.
    \end{align*}
\end{theorem}

\begin{corollary} \label{cor4}
    Under the assumptions of Theorem \ref{thm3}, assume further that $\eta=1$. That is, $\sigma>0$ is exponentially distributed with rate $r>0$. Then \eqref{gam_result} holds with 
    \begin{align*}
        \xi := \frac{B}{A} \frac{C_k^{(2n+1)/4}}{C^{(2m+1)/4}} > 0.
    \end{align*}
\end{corollary}
We remark on the exclusion of $k=0$ from Theorem \ref{thm3} above: Information about the nearest target determines the constants in \eqref{Ftau1}. Hence if $k=0$ the expressions in \eqref{Ftau1} and \eqref{Ftauk1} are identical and the asymptotic hitting probability is trivial. For this reason, $k=0$ is also omitted from the remaining results that follow.

Moreover, we note that in both Theorem \ref{thm3} and Corollary \ref{cor4} the dominant behavior of the asymptotic hitting probabilities depends only on $C$ and $C_k$. We show in section \ref{section3} that this dependence is in fact only on the geodesic lengths to the nearest target and to target $k$ for diffusive search. Also observe that as $r$ increases in Theorem \ref{thm3} and Corollary \ref{cor4}, the asymptotic hitting probabilities behave progressively more like decaying exponential functions. The same qualitative behavior is true in the following theorem, which is relevant to typical scenarios of diffusive search under uniformly distributed resetting.
\begin{theorem} \label{thm_uni}
    Under the assumptions of section \ref{prelim}, assume further that $\sigma>0$ is uniformly distributed on $[0,2/r]$ and for some $k\in\{1,\dots,K-1\}$,
    \begin{align*}
        F_{\tau}(t) &\sim At^me^{-C/t}\quad \text{as }t\to 0^+,\\
        F_{\tau}^{(k)}(t) &\sim Bt^ne^{-C_k/t}\quad \text{as }t\to 0^+,
    \end{align*}
    where $A>0$, $B>0$, $C_k>C>0$, and $m,\,n\in\mathbb{R}$. Then
    \begin{align*}
        \mathbb{P}(\kappa=k | \tau<\sigma) \sim \frac{B}{A} \frac{C}{C_k} \Big(\frac{r}{2}\Big)^{m-n} e^{-(C_k-C)r/2} \quad\text{as }r\to\infty.
    \end{align*}
\end{theorem}

In addition to instances of diffusive search, we consider superdiffusive search, discrete stochastic search on a network, and run-and-tumble search. Despite vast differences in the nature of these three search processes, the short-time behavior of their unconditional search times is similar. The following theorem describes the asymptotic behavior of the conditional hitting probabilities under exponential and gamma distributed resetting.

\begin{theorem} \label{net_gam}
    Under the assumptions of section \ref{prelim}, assume further that $\sigma>0$ is gamma distributed with shape $\eta>0$ and rate $r>0$ and that for some $k\in\{1,\dots,K-1\}$,
    \begin{align}
        F_{\tau}(t) &\sim At^m\quad \text{as }t\to 0^+, \label{7a}\\
        F_{\tau}^{(k)}(t) &\sim Bt^n\quad \text{as }t\to 0^+, \label{7b}
    \end{align}
    where $A>0$, $B>0$, and $0<m\leq n$. Then
    \begin{align*} 
        \mathbb{P}(\kappa=k | \tau<\sigma) \sim \frac{B}{A} \frac{\Gamma(n+\eta)}{\Gamma(m+\eta)} r^{m-n} \quad\text{as }r\to\infty.
    \end{align*}
\end{theorem}

The following corollary, while immediately clear from Theorem \ref{net_gam}, applies to superdiffusive search \cite{lawley2021super} and fundamentally differs from instances of diffusive search. In particular, these probabilities are neither zero nor one for every target in the asymptotic limit, hence, the nearest target may not be found.

\begin{corollary} \label{cor8}
    Under the assumptions of Theorem \ref{net_gam}, assume further that $m=n=1$. Then
    \begin{align*} 
        \mathbb{P}(\kappa=k | \tau<\sigma) \sim \frac{B}{A} \quad\text{as }r\to\infty.
    \end{align*}
\end{corollary}

The final theorem, stated below, holds for any short-time behavior of $F_{\tau}^{(k)}$ and $F_{\tau}$ specified in the aforementioned results and any other instance not specified herein. In particular, it conveys that hitting probabilities conditioned on a deterministic reset time behave like the ratio of these functions of $\tau>0$ evaluated at that time.

\begin{theorem} \label{sharp}
    Under the assumptions of section \ref{prelim}, assume further that $\sigma>0$ is deterministic and equal to $\sigma:=1/r$. Then
    \begin{align*} 
        \mathbb{P}(\kappa=k | \tau<\sigma) = \frac{F_{\tau}^{(k)}(1/r)}{F_{\tau}(1/r)} \sim \frac{\tilde{F}_{\tau}^{(k)}(1/r)}{\tilde{F}_{\tau}(1/r)} \quad\text{as }r\to\infty,
    \end{align*}
    where $\tilde{F}_{\tau}$ and $\tilde{F}_{\tau}^{(k)}$ denote the short-time behavior of $F_{\tau}$ and $F_{\tau}^{(k)}$, respectively,
    \begin{align*}
        F_{\tau}(t) &\sim \tilde{F}_{\tau}(t)\quad \text{as }t\to 0^+,\\
        F_{\tau}^{(k)}(t) &\sim \tilde{F}_{\tau}^{(k)}(t)\quad \text{as }t\to 0^+.
    \end{align*}
\end{theorem}

\section{Examples and numerical solutions} \label{section3}
The results stated in section \ref{section2} provide the exact asymptotic behavior of conditional hitting probabilities in the large $r$ limit in terms of the short-time behavior of $F_{\tau}$ and $F_{\tau}^{(k)}$. In this section we illustrate these results in several specific examples and compare our analytical results to numerical solutions.

\subsection{Pure diffusion in one dimension} \label{pd}
Consider a diffusive search process with diffusivity $D>0$ in one dimension. Assume the searcher starts at $x_0\in(0,\ell/2)$ where $V_0=(-\infty,0]$ and $V_1=[\ell,\infty)$ are targets. For ease of notation, define the distances from $x_0$ to each target by \[0<L_0 := x_0 < \ell-x_0 =: L_1.\]

It has been shown in numerous works, for example \cite{lawley2020dist}, that $F_{\tau}$ and $F_{\tau}^{(1)}$ satisfy \eqref{Ftau1} and \eqref{Ftauk1} where
\begin{align} \label{param}
    A = \sqrt{\frac{4D}{\pi L_0^2}},\quad B = \sqrt{\frac{4D}{\pi L_1^2}},\quad C = \frac{L_0^2}{4D},\quad C_1 = \frac{L_1^2}{4D},\quad m = n = 1/2.
\end{align}

\begin{figure}[h!]
  \centering
  \includegraphics[width=\textwidth]{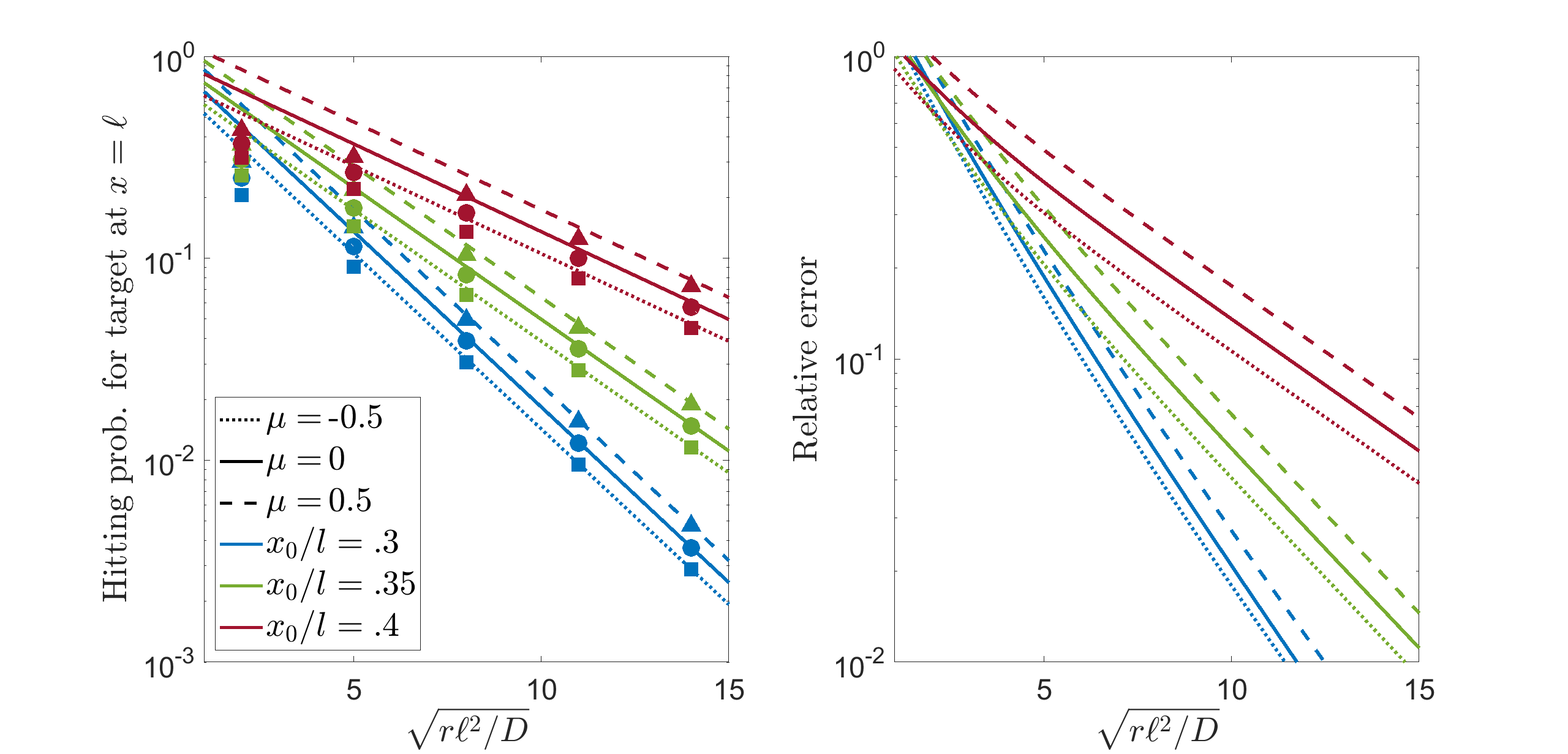}
 \caption{Conditional hitting probabilities for diffusion with diffusivity $D=1$, drift $\mu\in\mathbb{R}$, initial position $x_0\in(0,\ell/2)$, and exponentially distributed $\sigma>0$ with rate $r>0$. Curves illustrate analytical asymptotic estimates; markers denote numerical solutions. See subsections \ref{pd} and \ref{pd_drift} for details.}
 \label{fig1}
\end{figure}

\noindent
Hence, Corollary \ref{cor4} implies that for exponentially distributed $\sigma>0$ with rate $r>0$ the conditional hitting probabilities behave like 
\begin{align*}
    \mathbb{P}(\kappa=1 | \tau<\sigma) \sim e^{-\sqrt{\frac{(L_1-L_0)^2r}{D}}} = e^{-\sqrt{\frac{(\ell-2x_0)^2r}{D}}}\quad\text{as }r\to\infty.
\end{align*}
That is, the probability of the searcher finding the further target decays exponentially fast with respect to the square root of the resetting rate. In particular, notice that even if the searcher starts only slighter closer to one target, resetting can guarantee that the nearer target is found.

The asymptotic behavior of the conditional hitting probabilities for pure diffusion in the interval with numerous distributions for $\sigma$ are catalogued in Table \ref{table:1}. Figure \ref{fig1} illustrates these probabilities and compares them to numerical solutions.

\subsection{Diffusion with constant drift} \label{pd_drift}
Assuming the details of subsection \ref{pd}, consider in addition that the searcher experiences a constant drift $\mu\in\mathbb{R}$. That is, suppose that prior to a resetting event at time $\sigma>0$ the searcher position $\{X(t)\}_{t\geq 0}$ evolves according to a stochastic drift-diffusion process,
\begin{align*}
    \textup{d}X(t) = \mu \,\textup{d}t + \sqrt{2D}\,\textup{d}W(t),
\end{align*}
\noindent
where $\{W(t)\}_{t\geq 0}$ denotes a standard one-dimensional Brownian motion. It has again been previously shown \cite{Linn_2022} that $F_{\tau}$ and $F_{\tau}^{(1)}$ satisfy \eqref{Ftau1} and \eqref{Ftauk1} where
\[A = e^{-\mu L_0/2D}\sqrt{\frac{4D}{\pi L_0^2}},\quad B = e^{\mu L_1/2D}\sqrt{\frac{4D}{\pi L_1^2}},\]
and $C$, $C_1$, $m$, and $n$ in \eqref{param} are left unchanged.

Thus by Corollary \ref{cor4} where $\sigma>0$ is exponentially distributed with rate $r>0$, the conditional hitting probabilities behave like
\begin{align*}
    \mathbb{P}(\kappa=1 | \tau<\sigma) \sim e^{\frac{\mu\ell}{2D}} e^{-\sqrt{\frac{(\ell-2x_0)^2r}{D}}}\quad\text{as }r\to\infty.
\end{align*}

\begin{table}[h!]
\centering
\renewcommand{\arraystretch}{1.5}
\centerline{
\begin{tabular}{c c c c} 
 \hline \hline
   & Pure diffusion & With drift & With partially absorbing targets\\ [0.5ex] 
 \hline
 $\mathbb{P}(\kappa=1\,|\,\tau<\sigma_{\text{e}})$ & $e^{-\sqrt{(\ell-2x_0)^2r/D}}$ & $e^{-\sqrt{(\ell-2x_0)^2r/D} + (\mu\ell/2D)}$ & $\frac{\gamma_1}{\gamma_0}e^{-\sqrt{(\ell-2x_0)^2r/D}}$\\
  $\mathbb{P}(\kappa=1\,|\,\tau<\sigma_{\text{g}})$ & $\big(\frac{\ell-x_0}{x_0}\big)^{\eta-1} e^{-\sqrt{(\ell-2x_0)^2r/D}}$ & $\big(\frac{\ell-x_0}{x_0}\big)^{\eta-1} e^{-\sqrt{(\ell-2x_0)^2r/D} + (\mu\ell/2D)}$ & $\frac{\gamma_1}{\gamma_0} \big(\frac{\ell-x_0}{x_0}\big)^{\eta-1} e^{-\sqrt{(\ell-2x_0)^2r/D}}$\\ 
 $\mathbb{P}(\kappa=1\,|\,\tau<\sigma_{\text{u}})$ & $\frac{x_0^3}{(\ell-x_0)^3} e^{-(\ell-2x_0)r\ell/(8D)}$ & $\frac{x_0^3}{(\ell-x_0)^3} e^{-(\ell-2x_0)r\ell/(8D)+(\mu\ell/2D)}$ & $\frac{\gamma_1}{\gamma_0}\frac{x_0^4}{(\ell-x_0)^4} e^{-(\ell-2x_0)r\ell/(8D)}$\\ 
 $\mathbb{P}(\kappa=1\,|\,\tau<\sigma_{\text{s}})$ & $\frac{x_0}{\ell-x_0} e^{-(\ell-2x_0)r\ell/(4D)}$ & $\frac{x_0}{\ell-x_0} e^{-(\ell-2x_0)r\ell/(4D)+(\mu\ell/2D)}$ & $\frac{\gamma_1}{\gamma_0}\frac{x_0^2}{(\ell-x_0)^2} e^{-(\ell-2x_0)r\ell/(4D)}$\\ [1ex] 
 \hline \hline
\end{tabular}
}
\caption{Conditional hitting probabilities on the interval for types of diffusive search and various resetting distributions. Subscripts `e', `g', `u', and `s' on $\sigma>0$ denote exponential, gamma, uniform, and sharp distributions. See sections \ref{pd}-\ref{pd_abs} for details.}
\label{table:1}
\end{table}
\noindent
Qualitatively, we observe exponential decay of the hitting probability to the further target akin to the case of pure diffusion in subsection \ref{pd}. The sole difference in this case is a prefactor exhibiting the influence of drift, which increases the hitting probability of the target in its direction. Still, this influence is limited: in the limit as $r\to\infty$ the hitting probability to the nearest target goes to one independent of drift directionality.

The asymptotic conditional hitting probabilities for diffusion with drift in the interval with numerous distributions for $\sigma$ are catalogued in Table \ref{table:1}. Figure \ref{fig1} illustrates these probabilities and compares them to numerical solutions.

\subsection{Diffusion with partially absorbing targets} \label{pd_abs}
In the previous two examples, we assumed that the target was `found' by the searcher instantaneously upon arrival. That is, the so-called trapping rate was infinite. Alternatively, one may be interested in search processes wherein targets are found only after the searcher has spent sufficient time nearby. That is, for the unconditional search process with partially absorbing targets, the FPT is given by
\begin{align*}
    \tau_{\text{partial}} := \inf\{ t>0 : \lambda_k > \xi_k/\gamma_k\text{ for some }k\in\{0,\dots,K-1\} \}
\end{align*}
where $\{\xi_k\}_{k=0}^{K-1}$ are unit rate exponential random variables, $\{\gamma_k\}_{k=0}^{K-1}$ are non-negative trapping rates, and $\lambda_k$ is the local time of $X(t)$ on target $V_k$ \cite{grebenkov2020}.

Now in addition to the details of subsection \ref{pd}, suppose the trapping rate of the target at $x=0$ is given by $\gamma_0\in(0,\infty)$ and that of $x=\ell$ is $\gamma_1\in(0,\infty)$. To be precise, suppose the survival probability for the search process, whose backward Kolmogorov equation is given by the one-dimensional heat equation,
\begin{align*}
    \partial_t S = D\partial_{xx} S, \quad x\in(0,\ell),
\end{align*}
with initial condition $S(x,t=0) = 1$, has Robin boundary conditions,
\begin{align*}
    D\partial_x S &= \gamma_0 S,\quad x = 0,\\
    -D\partial_x S &= \gamma_1 S,\quad x = \ell.
\end{align*}

Solving this system and taking the limits as right and left boundaries of the interval extend out to infinity yield the short time behavior of $F_{\tau}$ and $F_{\tau}^{(1)}$, respectively. These expressions have been shown \cite{lawley2020dist} to satisfy \eqref{Ftau1} and \eqref{Ftauk1} with \[A = \frac{2\gamma_0}{L_0}\sqrt{\frac{4D}{\pi L_0^2}},\quad B = \frac{2\gamma_1}{L_1}\sqrt{\frac{4D}{\pi L_1^2}},\quad m = n = \frac{3}{2}\]
with $C$ and $C_1$ unchanged from \eqref{param}.

\begin{figure}[h!]
  \centering
  \includegraphics[width=\textwidth]{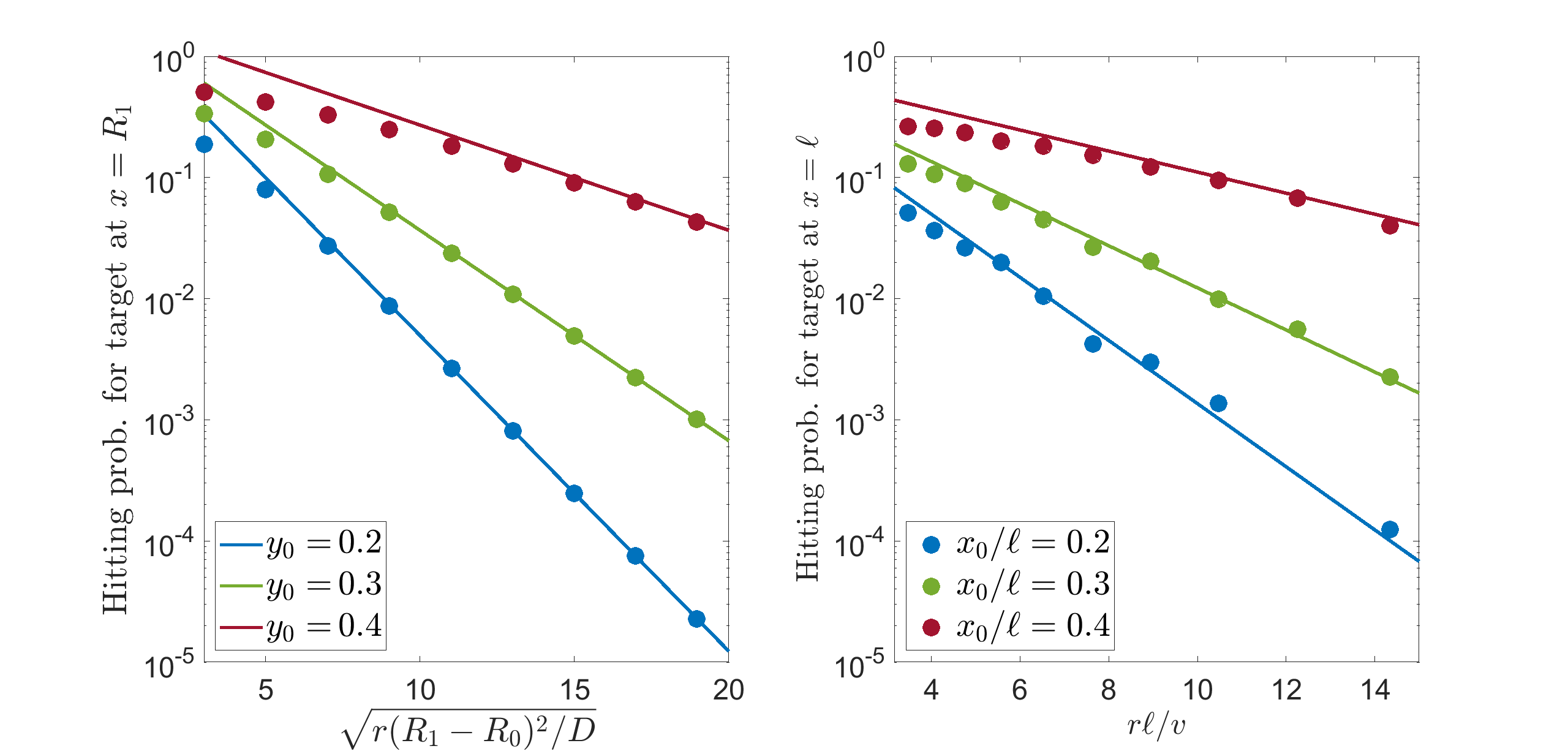}
 \caption{Left: Conditional hitting probabilities for diffusion with diffusivity $D>0$ between concentric spheres of radii $R_0$ and $R_1$ and exponentially distributed $\sigma>0$ with rate $r>0$. Above, $y_0:=(x_0-R_0)/(R_1-R_0)$. Curves illustrate analytical asymptotic estimates; markers denote numerical solutions. See subsection \ref{conc_sph} for details. Right: Conditional hitting probabilities for an RTP with velocity $v>0$, tumble rate $\lambda>0$, and exponentially distributed $\sigma>0$ with rate $r>0$. Curves illustrate analytical asymptotic estimates; markers denote simulations. See subsection \ref{RTP} for details.}
 \label{fig2}
\end{figure}

Corollary \ref{cor4} thus implies that, for exponentially distributed $\sigma>0$ with rate $r>0$, the conditional hitting probabilities behave like
\begin{align*}
    \mathbb{P}(\kappa=1 | \tau<\sigma) \sim \frac{\gamma_1}{\gamma_0}e^{-\sqrt{\frac{(\ell-2x_0)^2r}{D}}} \quad\text{as }r\to\infty.
\end{align*}

\noindent
Qualitatively, we observe exponential decay of the hitting probability to the further target again akin to the case of pure diffusion in subsection \ref{pd}. The prefactor in this case illustrates the relationship between the trapping rates at each target: a relatively high trapping rate at the further target increases its corresponding hitting probability but this influence is always exceeded by the exponential decay in the limit as $r\to\infty$.

The asymptotic conditional hitting probabilities for pure diffusion with partially absorbing targets in the interval with numerous distributions for $\sigma$ are catalogued in Table \ref{table:1}.

\subsection{Diffusion between concentric spheres} \label{conc_sph}
Consider a diffusive search process with diffusivity $D>0$ in three spatial dimensions between two perfectly absorbing concentric spheres of radii $R_1>R_0>0$. Let the initial position of the searcher be $x_0 := \|X(0)\| \in (R_0,(R_0+R_1)/2)$ where $\|\cdot\|$ denotes the Euclidean norm. In words, suppose the searcher starts closer to the inner sphere. For ease of notation, define the distances from $x_0$ to each target by
\begin{align*}
    0 < L_0 := x_0 - R_0 < R_1 - x_0 =: L_1.
\end{align*}

\begin{table}[h!]
\centering
\renewcommand{\arraystretch}{1.5}
\centerline{
\begin{tabular}{c c c c} 
 \hline \hline
   & Diffusion between concentric spheres & RW on a network & RTP in an interval\\ [0.5ex] 
 \hline
 $\mathbb{P}(\kappa=1\,|\,\tau<\sigma_{\text{e}})$ & $\frac{R_1}{R_0} e^{-\sqrt{(R_0+R_1-2x_0)^2r/D}}$ & $\frac{\Lambda_1}{\Lambda_0}r^{m-n}$ & $\frac{1-q}{q} e^{-(\lambda+r)(\frac{L_1}{v_1}-\frac{L_0}{v_0})}$ \\
  $\mathbb{P}(\kappa=1\,|\,\tau<\sigma_{\text{g}})$ & $\Big(\frac{R_1-x_0}{x_0-R_0}\Big)^{\eta-1} e^{-\sqrt{(R_0+R_1-2x_0)^2r/D}}$ & $\frac{\Gamma(n+\eta)m!\Lambda_1}{\Gamma(m+\eta)n!\Lambda_0} r^{m-n}$ & $\frac{1-q}{q}\left( \frac{L_1v_0}{L_0v_1} \right)^{\eta-1} e^{-(\lambda+r)(\frac{L_1}{v_1}-\frac{L_0}{v_0})}$ \\ 
 $\mathbb{P}(\kappa=1\,|\,\tau<\sigma_{\text{u}})$ & $\frac{R_1(x_0-R_0)^3}{R_0(R_1-x_0)^3}e^{-r(R_0+R_1-2x_0)(R_1-R_0)/(8D)}$ & $\frac{(m+1)!\Lambda_1}{(n+1)!\Lambda_0} \big(\frac{r}{2}\big)^{m-n}$ & \\ 
 $\mathbb{P}(\kappa=1\,|\,\tau<\sigma_{\text{s}})$ & $\frac{R_1(x_0-R_0)}{R_0(R_1-x_0)}e^{-r(R_0+R_1-2x_0)(R_1-R_0)/(4D)}$ & $\frac{m!\Lambda_1}{n!\Lambda_0} r^{m-n}$ & \\ [1ex] 
 \hline \hline
\end{tabular}
}
\caption{Conditional hitting probabilities for various types of stochastic search and resetting distributions. Subscripts `e', `g', `u', and `s' on $\sigma>0$ denote exponential, gamma, uniform, and sharp distributions. See sections \ref{conc_sph}-\ref{RTP} for details.}
\label{table:2}
\end{table}

It has been derived in numerous works, for example \cite{Linn_2022}, that $F_{\tau}$ and $F_{\tau}^{(1)}$ satisfy \eqref{Ftau1} and \eqref{Ftauk1} where
\[A = \frac{R_0}{x_0}\sqrt{\frac{4D}{\pi L_0^2}},\quad B = \frac{R_1}{x_0}\sqrt{\frac{4D}{\pi L_1^2}},\quad C = \frac{L_0^2}{4D},\quad C_1 = \frac{L_1^2}{4D},\quad m = n = 1/2.\]
Hence, by Corollary \ref{cor4} wherein $\sigma>0$ is assumed to be exponentially distributed with rate $r>0$, the conditional hitting probabilities behave like
\begin{align*}
    \mathbb{P}(\kappa=1 | \tau<\sigma) \sim \frac{R_1}{R_0} e^{-\sqrt{\frac{(R_0+R_1-2x_0)^2r}{D}} } \quad\text{as }r\to\infty.
\end{align*}

The asymptotic conditional hitting probabilities for diffusion between concentric 3D spheres with numerous distributions for $\sigma$ are catalogued in Table \ref{table:2}. Figure \ref{fig2} illustrates these probabilities and compares them to numerical solutions.

\subsection{Random walk on a network} \label{RW_net}
Now consider the spatially discrete stochastic search process of a random walk on a discrete network. To be precise, let $\{X(t)\}_{t\geq 0}$ be a continuous-time Markov chain on a finite or countably infinite state space, $I$. We assume the jump rates are bounded to preclude the possibility of infinitely many jumps occurring in a short time on a countably infinite network. Denote by $i_0\in I$ the initial position of $X(t)$ and denote target nodes (or sets of nodes that constitute targets) by $V_k\in I$ where $i_0 \not\in \cup_{k=0}^{K-1} V_k$.

Previous work on stochastic search processes on networks implies that $F_{\tau}$ and $F_{\tau}^{(k)}$ satisfy \eqref{7a} and \eqref{7b} where
\begin{align*}
    A = \frac{\Lambda_0}{m!}, \quad B = \frac{\Lambda_k}{n!}
\end{align*}
with $\Lambda_0$ and $\Lambda_k$ denoting the products of the jump rates of $X(t)$ along the shortest path from $i_0$ to the nearest target and from $i_0$ to target $V_k$, respectively \cite{lawley2020networks}. Further, $n\geq 1$ is the minimum number of jumps between $i_0$ and $V_k$ and similarly $m\geq 1$ is the minimum number of jumps between $i_0$ and the nearest target. If there are multiple shortest paths, $\Lambda_0$ or $\Lambda_k$ are the sums of the products of the jump rates along these paths.

\begin{figure}[h!]
  \centering
  \includegraphics[width=\textwidth]{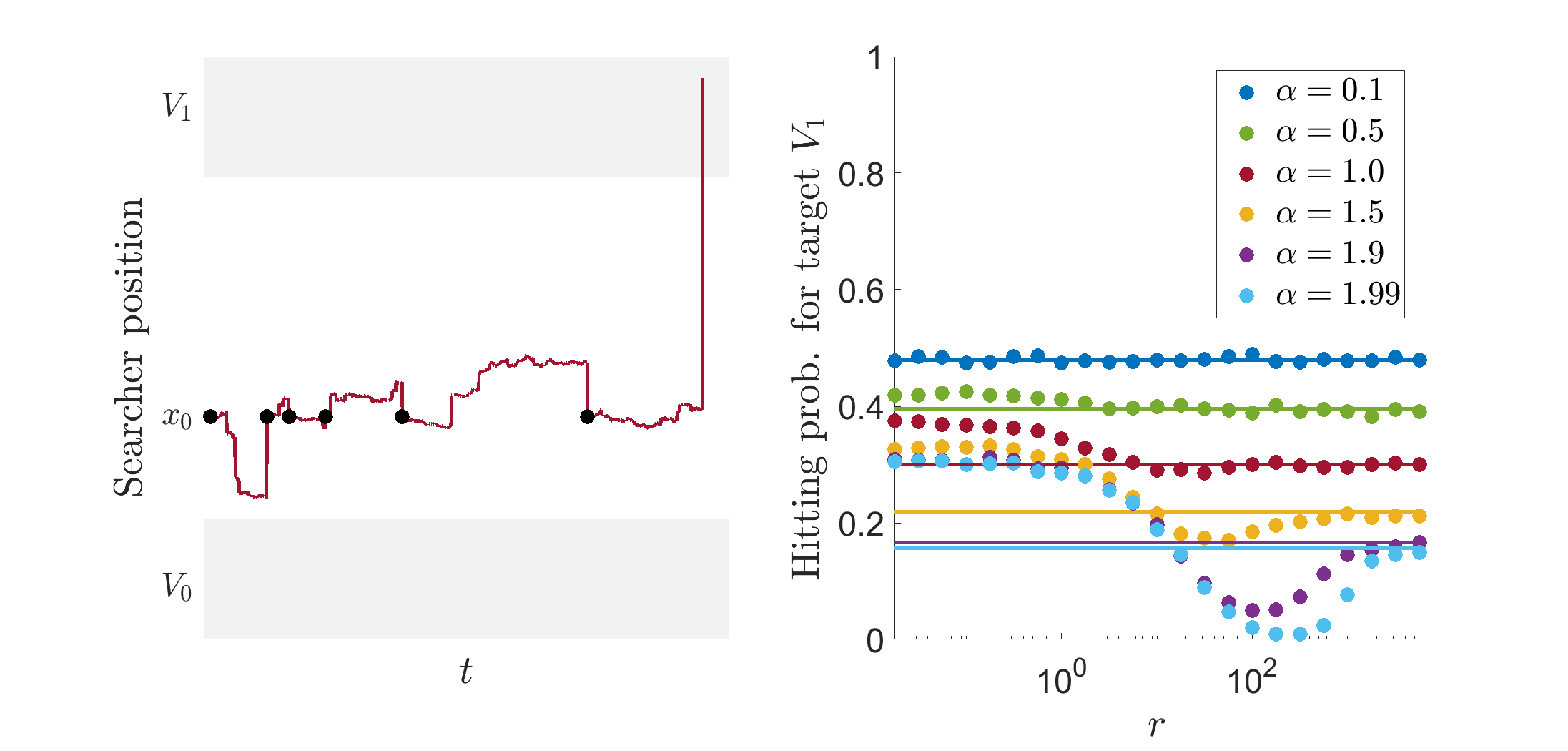}
 \caption{Left: Path of an unconditional $\alpha = 1$ L\'evy flight search. Markers denote resetting. Right: Conditional hitting probabilities for L\'evy flights with $\alpha\in(0,2)$ and exponentially distributed $\sigma>0$ with rate $r>0$. Lines illustrate analytical asymptotic estimates; markers denote simulations. See subsection \ref{superdiff} for details.}\label{levy}
\end{figure}

Theorem \ref{net_gam} thus implies that, for exponentially distributed $\sigma>0$ with rate $r>0$, the conditional hitting probabilities behave like
\begin{align*}
    \mathbb{P}(\kappa=k | \tau<\sigma) \sim \frac{\Lambda_k}{\Lambda_0} r^{m-n} \quad\text{as }r\to\infty.
\end{align*}

The asymptotic conditional hitting probabilities for random walks on networks with numerous distributions for $\sigma$ are catalogued in Table \ref{table:2}.

\subsection{Superdiffusive L\'evy flight in one dimension} \label{superdiff}
Suppose now the stochastic search process is a superdiffusive L\'evy flight starting from $x_0\in(0,\ell/2)$ where $V_0=(-\infty,0]$ and $V_1=[\ell,\infty)$ are targets.

A useful characterization of a superdiffusive L\'evy flight in this setting is a subordinated Brownian motion, which is a random time change of a standard Brownian motion. In particular, if $\{B(s)\}_{s\geq 0}$ is a standard Brownian motion with unit diffusivity and $\{S(t)\}_{t\geq 0}$ is a non-decreasing L\'evy process with $S(0)=0$. Denote the path of the search process by
\begin{align*}
    X(t) := B(S(t)) + x_0,\quad t\geq 0.
\end{align*}

In a particular scaling limit, the probability density $p(x,t)$ of the position of the search process satisfies
\begin{align*}
    \partial_t p = -\bar{D}(-\Delta)^{\alpha/2}p,\quad t>0,
\end{align*}
where $\bar{D}>0$ is the generalized diffusivity and $(-\Delta)^{\alpha/2}$ is the fractional Laplacian with $\alpha\in(0,2)$ \cite{lischke2020}. Assuming $x_0\notin V_0\cup V_1$, one can show that $F_{\tau}$ and $F_{\tau}^{(1)}$ satisfy \eqref{7a} and \eqref{7b} with
\begin{align*}
    A = \frac{\Gamma(\alpha)\bar{D} \sin(\alpha\pi/2)}{\pi (\ell-x_0)^{\alpha}} + \frac{\Gamma(\alpha)\bar{D} \sin(\alpha\pi/2)}{\pi x_0^{\alpha}} ,\quad B = \frac{\Gamma(\alpha)\bar{D} \sin(\alpha\pi/2)}{\pi (\ell-x_0)^{\alpha}}
\end{align*}
and $m=n=1$ \cite{palyulin2019,lawley2021super}. Hence, by Corollary \ref{cor8} with gamma distributed $\sigma>0$ with shape $\eta>0$ and rate $r>0$, the conditional hitting probabilities behave like
\begin{align*}
    \mathbb{P}(\kappa=1 | \tau<\sigma) \sim \frac{x_0^{\alpha}}{x_0^{\alpha} + (\ell-x_0)^{\alpha}}\quad\text{as }r\to\infty.
\end{align*}

This result differs substantially from the previously considered examples; not only is the asymptotic hitting probability independent of the (generalized) diffusivity but it is also independent of the resetting rate. That is, regardless of the search speed, there is strictly positive probability that the search finds a far target in lieu of the nearest one.

\subsection{Run and tumble in one dimension} \label{RTP}
We conclude our examples with an application to a stochastically resetting run and tumble particle (RTP). Recent studies on this process concern steady-state behavior in the absence of targets \cite{paul_RTP,KS_RTP} and first passage times in the presence of targets \cite{Tucci_RTP}. Here we consider the hitting probabilities of an RTP in one dimension between targets at $V_0=(-\infty,0]$ and $V_1=[\ell,\infty)$ and initial position $x_0\in(0,\ell)$. We define $L_0:=x_0$ and $L_1:=\ell-x_0$. The searcher has probability $q\in[0,1]$ of initially moving in the negative x-direction and switches between velocities $v_1>0$ and $-v_0<0$ at Poissonian rate $\lambda>0$. Moreover, we assume $L_0/v_0 < L_1/v_1$. Ultimately by way of Theorem \ref{net_gam}, we determine that, for exponentially distributed $\sigma>0$ with rate $r>0$, the conditional hitting probabilities behave like
\begin{align*}
    \mathbb{P}(\kappa = 1 | \tau < \sigma) \sim \frac{1-q}{q} e^{-(\lambda+r)(L_1/v_1-L_0/v_0)}\quad\text{as }r\to\infty.
\end{align*}

The details of this calculation are contained in section \ref{rrttpp} of the Appendix. Moreover, the asymptotic conditional hitting probabilities for one-dimensional RTPs with numerous distributions for $\sigma$ are catalogued in Table \ref{table:2}.  Figure \ref{fig2} illustrates these probabilities and compares them to numerical solutions.

\section{Application to a decision-making model} \label{section3b}
In this section, we apply Theorem \ref{sharp} to an evidence accumulation model for decision making, which is a model type that is widely used to describe how individuals gather information to make decisions \cite{acemoglu2011,karamched2020prl,karamched2020,reina2023,linn2024fast}. In particular, it has been shown that the log-likelihood ratio describing the belief state of an individual can be approximated by a drift-diffusion process \cite{busemeyer1993,bogacz2006, banburismus2002}. Within this framework we determine the asymptotic behavior of an individual who makes a decision before a fast deterministic deadline. In the fast decision limit, we show that the individual decides in agreement with her initial bias. How this limit depends on the `tightness' of the deadline is detailed below.

Consider an individual gathering information according to 
\begin{align*}
    \text{d}X(t) = \mu\, \text{d}t + \sqrt{2D}\,\text{d}W(t),
\end{align*}
where the drift parameter, $\mu\in\mathbb{R}$, represents bias toward a particular decision and the diffusivity, $D>0$, denotes the magnitude of the Brownian motion, $\{W(t)\}_{t\geq 0}$, which describes the noisiness in the information gathering process. Suppose first that the individual is deciding between two choices and that a choice is made when the drift-diffusion process reaches one of two thresholds, $-\theta$ or $\theta$, for fixed $\theta>0$. Suppose further that the individual has an initial bias of $x_0\in(0,\theta)$.

\begin{figure}[h!]
  \centering
  \includegraphics[width=\textwidth]{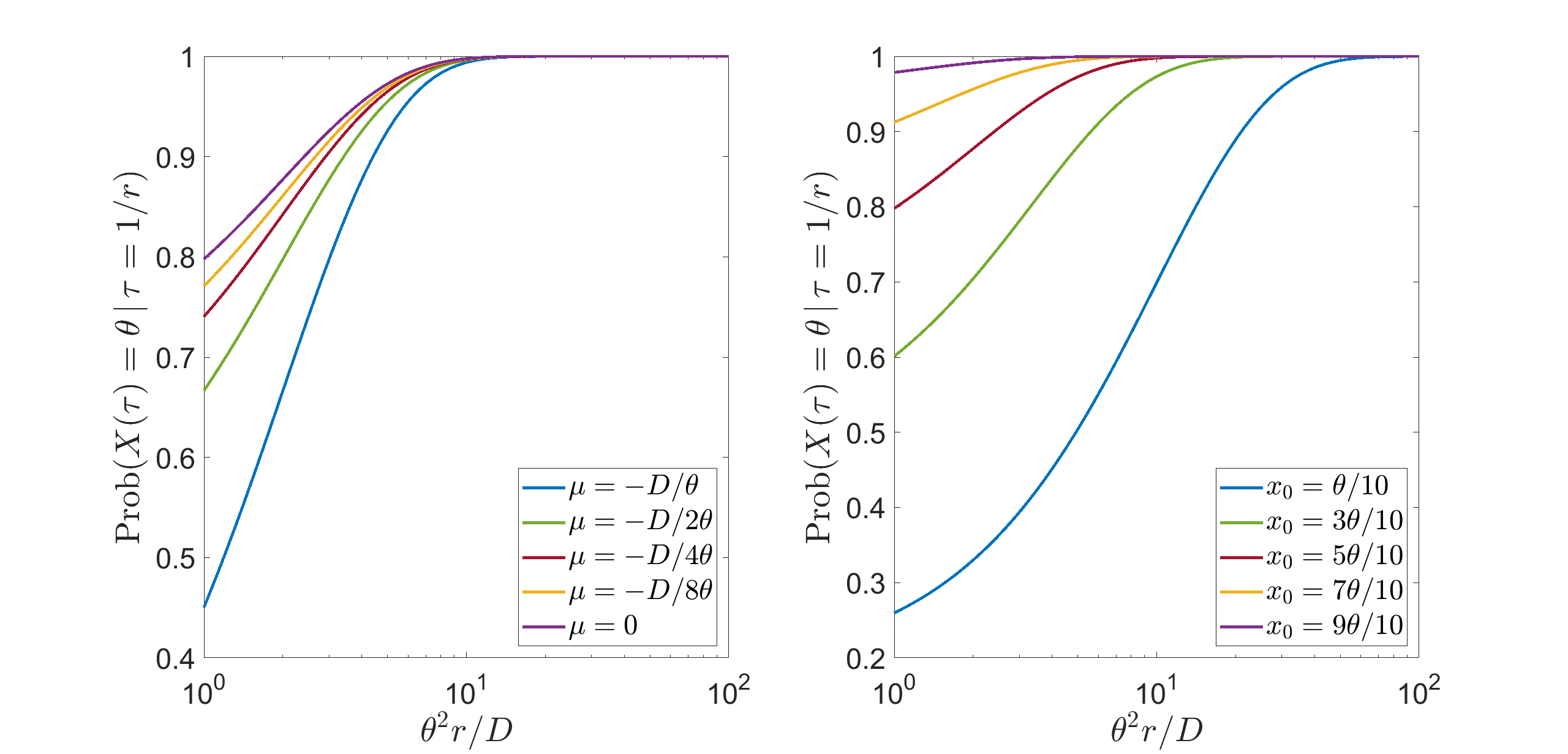}
  \caption{Left: The probability that an individual decides in agreement with her initial bias converges to unity in the tight deadline limit independent of an external bias. In this panel, $x_0=\theta/2$. Right: The probability that an individual decides in agreement with her initial bias converges to unity in the tight deadline limit independent of initial bias strength. In this panel, $\mu=0$.}
 \label{fig3}
\end{figure}

Under these assumptions, we inquire about the decision made by the individual before a sharp (deterministic) deadline $\sigma_s \equiv 1/r$ for fixed $r>0$. We emphasize that in this setting the individual has no knowledge of the deadline but nevertheless decides quickly. By Theorem \ref{sharp}, where $\tilde{F}_{\tau}$ and $\tilde{F}^{(k)}_{\tau}$ given in subsection \ref{pd_drift} with $L_0 = \theta-x_0$, $L_k = \theta+x_0$,
\begin{align*}
    A = e^{-\mu L_0/2D}\sqrt{\frac{4D}{\pi L_0^2}},\quad B = e^{\mu L_k/2D}\sqrt{\frac{4D}{\pi L_k^2}},\quad C = \frac{L_0^2}{4D},\quad C_k = \frac{L_k^2}{4D},
\end{align*}
and $m=n=1/2$, we determine that the probability of the individual deciding against its initial bias is
\begin{align} \label{lim}
    \mathbb{P}(X(\tau)=-\theta\,|\,\tau=1/r) \sim \frac{\theta-x_0}{\theta+x_0} e^{-\theta\mu/D} e^{-\theta x_0r/D}\quad\text{as }r\to\infty.
\end{align}

Qualitatively, we infer that quick decisions are biased decisions. That is, the individual will always decide in agreement with her initial bias in the tight deadline limit, even in the presence of an external bias `pushing' the agent toward the alternative. That is, regardless of external bias,
\begin{align*}
    \mathbb{P}(X(\tau)=\theta\,|\,\tau=1/r) \to 1\quad\text{as }r\to\infty. 
\end{align*}

While such hasty decisions may not be made frequently by a single agent, first deciders among groups of agents are reliably quick and similarly biased \cite{linn2024fast}. The influence of these early decisions on agents still deliberating can be substantial; recent work suggests that having many fast deciders can sway the group to the same hasty decision, but having only a few fast deciders can caution remaining agents against biased decision-making \cite{karamched2020prl}.

Unsurprisingly, the presence of a disagreeable external drift does slow the convergence to zero of the probability in \eqref{lim}. We illustrate this behavior in the left panel of Figure \ref{fig3} to emphasize its limited influence. Relatively weak initial biases also slow the convergence to zero of the probability in \eqref{lim} but its influence is similarly limited. We illustrate this behavior in the right panel of Figure \ref{fig3}. Finally we remark that these results extend to higher dimensional landscapes for decision making. In these settings, the lengths $L_0$ and $L_k$ denote the geodesic lengths to the nearest decision threshold and another decision threshold of interest.

\section{Discussion} \label{section5}
In this work, we determine hitting probabilities for fast stochastic search. Our results yield the exact asymptotics of these hitting probabilities in terms of the short-time behavior of the stochastic search process without a time constraint. In typical scenarios of diffusive search the limit of these asymptotic quantities reveals that the nearest target is always found. In particular, hitting probabilities to far targets exhibits exponential decay. This behavior readily breaks down for certain non-diffusive search processes including random walks on networks and superdiffusive search. In the random walk case, hitting probabilities to far targets exhibit polynomial decay. Superdiffusive search, however, can always find far targets quickly. That is, all limiting hitting probabilities associated with superdiffusive search are strictly positive. Hence, with sufficiently frequent resetting, modulation of the search process and various system parameters are an effective means of `sharpening' hitting probabilities to achieve a specific goal.

This work adds a new piece to the vast and growing body of work on stochastic search processes and, in particular, fast stochastic search. Until recently, the remarkable acuity of timing and directionality in biological processes was poorly understood. The so-called redundancy principle has since provided an explanation for processes in which many nearly identical copies of bio-entities, though effectively independent, work together to increase the likelihood of fast, precise, and accurate action \cite{schuss2019}. In the case of $N\gg 1$ independent searchers, the moments of the fastest FPT for typical scenarios of diffusive search are given by \cite{lawley2020uni}
\begin{align*}
    \mathbb{E}[T_N^m] \sim \left(\frac{L_0^2}{4D\ln N}\right)^m\quad\text{as }N\to\infty,
\end{align*}
where $D>0$ is the diffusivity and $L_0>0$ is a geodesic length to the nearest target. The associated hitting probabilities to far targets are \cite{Linn_2022}
\begin{align*}
    \mathbb{P}(\kappa_N = k) \sim \beta(\ln N)^{((L_k/L_0)-1)/2}N^{1-(L_k/L_0)^2}\quad\text{as }N\to\infty
\end{align*}
where $\beta>0$ is a constant dependent on system parameters and $L_k>0$ is a geodesic length to a far target. In the case of evidence accumulation models for decision making, the analogous hitting probability result indicates that the fastest decider is she who had the strongest initial bias \cite{linn2024fast}.

Where redundancy is not at play, or perhaps working simultaneously with redundancy, stochastic inactivation or resetting serves a similar biological purpose. Recent work in this area characterizes when infrequent stochastic resetting increases the hitting probability or decreases the mean FPT to particular targets \cite{PhysRevE.102.022115}. For typical scenarios of diffusive search, the moments of the FPT under frequent exponentially distributed resetting with rate $r\gg L_0^2/4D$ is \cite{FPTuFSR}
\begin{align*}
    \mathbb{E}[T_{reset}^m] \sim \left(\frac{e^{L_0\sqrt{r/D}}}{r}\right)^m\quad\text{as }r\to\infty.
\end{align*}
Moreover, the same quantity for an inactivating searcher is \cite{lawley2021mortal}
\begin{align*}
    \mathbb{E}[T_{inact}^m ] \sim \left(\frac{L_0^2}{4Dr}\right)^{m/2}\quad\text{as }r\to\infty.
\end{align*}
The associated hitting probabilities to far targets, as derived in this work, are
\begin{align*}
    \mathbb{P}(\kappa_r = k) \sim \eta e^{-(L_k-L_0)\sqrt{r/D}} \quad\text{as }r\to\infty
\end{align*}
where $\eta>0$ is a constant depending on system parameters. In terms of the decision making model discussed in section \ref{section3b}, the analogous hitting probability results indicate that fast decisions agree with an individual's initial bias. There have been numerous recent studies quantifying the `cost' associated with redundancy, resetting, and inactivation and analyzing its effects \cite{Bhat2016,talfriedman2020, debruyne2023,Sunil2023}. While we do not consider such a cost in this work, we acknowledge the potential influence on system behavior and hence its importance for certain biological and decision-making applications.

To conclude, we have shown that the hitting probabilities of fast stochastic search can differ tremendously from those free of time constraints. These results build on the growing body of work seeking to understand the nature of stochastic search in biology, sociology, and other scientific disciplines. In particular, we show that while hitting probabilities for typical scenarios of diffusive search decay exponentially fast in the frequent resetting limit, other search processes exhibit notable quantitative and qualitative differences. It is also the case, as highlighted in this work, that while hitting probabilities for unconstrained stochastic search can prove challenging to compute, especially in complicated domains, hitting probabilities for fast stochastic search rely on very few details of the system and can often be easily determined.

\newpage
\section{Appendix}
This appendix contains proofs of the propositions and theorems stated in section \ref{section2} as well as details of the numerical methods used in section \ref{section3}.

\subsection{Proofs}
\begin{proof}[Proof of Proposition \ref{prop3}]
    The proof of this proposition is in the Appendix of \cite{FPTuFSR}. We repeat it here for readability: The exponential term in the integrand of \eqref{prop3eq} achieves its maximum at
    \begin{align*}
        t^* := \sqrt{C/r}.
    \end{align*}
    Let $r>0$ be sufficiently large so that $t^*\in(0,\delta)$. Under the change of variables $s=t/t^*$, the integral in \eqref{prop3eq} becomes
    \begin{align} \label{laplace1}
        \int_0^{\delta} e^{-rt} t^b e^{-C/t} \,\textup{d}t = \left(\frac{C}{r}\right)^{\frac{b+1}{2}} \int_0^{\delta/t^*} s^b \text{exp}\left(-\sqrt{Cr}\left(s + s^{-1} \right) \right)\,\textup{d}t.
    \end{align}
    Applying Laplace's method to \eqref{laplace1} where $s=1\in(0,\delta/t^*)$ corresponds to the maximum of the exponential term in the integrand completes the proof. 
\end{proof}

\noindent
\begin{proof}[Proof of Theorem \ref{thm3}]
    Let $\epsilon>0$ and define
    \begin{align} \label{Iab}
        I_{a,b} := \int_a^b S_{\sigma}(t)\,\textup{d}F_{\tau}(t).
    \end{align}
    By the assumption in \eqref{Ftau1}, there exists $\delta>0$ so that
    \begin{align*}
        (1-\epsilon)At^me^{-C/t} \leq F_{\tau}(t) \leq (1+\epsilon)At^me^{-C/t}\quad\text{for all }t\in(0,\delta).
    \end{align*}
    Under this assumption, we can integrate by parts and bound $I_{0,\delta}$ from above,
    \begin{align} \label{t3eq}
    \begin{split}
        I_{0,\delta} = S_{\sigma}(\delta) &F_{\tau}(\delta       ) - \int_0^{\delta} F_{\tau}(t)\,\textup{d}S_{\sigma}(t) = S_{\sigma}(\delta) F_{\tau}(\delta) + \int_0^{\delta} F_{\tau}(t) f_{\sigma}(t)\,\textup{d}t \\
        &\leq \left( 1 - \frac{\gamma(\eta,r\delta)}{\Gamma(\eta)} \right) F_{\tau}(\delta) + \frac{(1+\epsilon)Ar^{\eta}}{\Gamma(\eta)} \int_0^{\delta} t^{m+\eta-1}e^{-(C/t)-rt}\,\textup{d}t,
    \end{split}
    \end{align}
    where $\gamma(s,x) := \int_0^x t^{s-1} e^{-t}\,\textup{d}t$ denotes the lower incomplete gamma function. The first term on the right-hand side of \eqref{t3eq} decays exponentially fast as $r\to\infty$. To control $I_{\delta,\infty}$, we use integration by parts to write
\begin{align*}
    I_{\delta,\infty} &= -S_{\sigma}(\delta) F_{\tau}(\delta) + \int_{\delta}^{\infty} F_{\tau}(t)\,\text{d}F_{\sigma}(t)\nonumber\\
    &\leq -S_{\sigma}(\delta) F_{\tau}(\delta) + \int_{\delta}^{\infty} \,\text{d}F_{\sigma}(t)\nonumber\\
    &= S_{\sigma}(\delta) (1-F_{\tau}(\delta)),
\end{align*}
which decays at least exponentially fast as $r\to\infty$. We emphasize that this bound holds in general so long as $\tau$ and $\sigma$ satisfy the assumptions in section \ref{prelim}.  Hence, recalling the definition of $p$ in \eqref{pdef} and applying Proposition \ref{prop3} yields
    \begin{align*}
        \limsup_{r\to\infty} \frac{p}{(A/\Gamma(\eta)) \sqrt{\pi} C^{(2m+1)/4)} r^{(4\eta-2m-3)/4} e^{-\sqrt{4Cr}}} \leq 1 + \epsilon.
    \end{align*}
    An analogous argument gives the lower bound
    \begin{align*}
        \liminf_{r\to\infty} \frac{p}{(A/\Gamma(\eta)) \sqrt{\pi} C^{(2m+1)/4)} r^{(4\eta-2m-3)/4} e^{-\sqrt{4Cr}}} \geq 1 - \epsilon.
    \end{align*}
    Repeating this procedure for $p_k$ as defined in \eqref{pkdef} yields similar relations,
    \begin{align*}
        \limsup_{r\to\infty} \frac{p_k}{(B/\Gamma(\eta)) \sqrt{\pi} C_k^{(2n+1)/4)} r^{(4\eta-2n-3)/4} e^{-\sqrt{4C_kr}}} \leq 1 + \epsilon,\\
        \liminf_{r\to\infty} \frac{p_k}{(B/\Gamma(\eta)) \sqrt{\pi} C_k^{(2n+1)/4)} r^{(4\eta-2n-3)/4} e^{-\sqrt{4C_kr}}} \geq 1 - \epsilon.
    \end{align*}
    Since $\epsilon\in(0,1)$ is arbitrary, recalling that
    \begin{align*}
        \mathbb{P}(\kappa=k | \tau<\sigma) = \frac{p_k}{p} = \frac{\int_0^{\infty} S_{\sigma}(t)\,\textup{d}F^{(k)}_{\tau}(t)}{\int_0^{\infty} S_{\sigma}(t)\,\textup{d}F_{\tau}(t)}
    \end{align*}
    completes the proof.
\end{proof}

\noindent
\begin{proof}[Proof of Theorem \ref{thm_uni}]
    Let $\epsilon>0$. From \eqref{Ftau1}, there exists $\delta>0$ so that
    \begin{align*}
        (1-\epsilon)At^me^{-C/t} \leq F_{\tau}(t) \leq (1+\epsilon)At^me^{-C/t}\quad\text{for all }t\in(0,\delta).
    \end{align*}
    Let $r>0$ be sufficiently large so that $2/r<\delta$ and let $I_{a,b}$ be as defined in \eqref{Iab}. Integrating by parts to bound $I_{0,2/r}$ from above, we have
    \begin{align} \label{02r}
    \begin{split}
        I_{0,2/r} &= -\int_0^{2/r} F_{\tau}(t)\,\textup{d}S_{\sigma}(t) = \int_0^{2/r} F_{\tau}(t)f_{\sigma}(t)\,\textup{d}t\\
        &\leq \frac{(1+\epsilon)Ar}{2} \int_0^{2/r} t^m e^{-C/t} \,\textup{d}t = \frac{(1+\epsilon)Ar}{2} C^{m+1} \Gamma(-m-1,Cr/2)
    \end{split}
    \end{align}
    where $\Gamma(s,x):= \int_x^{\infty} t^{s-1} e^{-t}\,\textup{d}t$ denotes the upper incomplete gamma function. Applying the asymptotic relation $\Gamma(s,x) \sim x^{s-1} e^{-x}$ as $x\to\infty$ to the rightmost side of \eqref{02r} and observing from the proof of Theorem \ref{thm3} that $I_{2/r,\infty}$ decays exponentially fast, we infer that
    \begin{align*}
        \limsup_{r\to\infty} \frac{p}{(A/C)(r/2)^{-m-1}e^{-Cr/2}} \leq 1 + \epsilon.
    \end{align*}
    An analogous argument gives the lower bound
    \begin{align*}
        \liminf_{r\to\infty} \frac{p}{(A/C)(r/2)^{-m-1}e^{-Cr/2}} \geq 1 - \epsilon.
    \end{align*}
    Repeating this procedure for $p_k$ yields similar relations,
    \begin{align*}
        \limsup_{r\to\infty} \frac{p_k}{(B/C_k)(r/2)^{-n-1}e^{-C_kr/2}} \leq 1 + \epsilon,\\
        \liminf_{r\to\infty} \frac{p_k}{(B/C_k)(r/2)^{-n-1}e^{-C_kr/2}} \geq 1 - \epsilon.
    \end{align*}
    Since $\epsilon\in(0,1)$ is arbitrary, recalling that
    \begin{align*}
        \mathbb{P}(\kappa=k | \tau<\sigma) = \frac{p_k}{p}  = \frac{\int_0^{\infty} S_{\sigma}(t)\,\textup{d}F^{(k)}_{\tau}(t)}{\int_0^{\infty} S_{\sigma}(t)\,\textup{d}F_{\tau}(t)}
    \end{align*}
    completes the proof.
\end{proof}

\noindent
\begin{proof}[Proof of Theorem \ref{net_gam}]
    Let $\epsilon>0$. From \eqref{Ftau1}, there exists $\delta>0$ so that
    \begin{align*}
        (1-\epsilon)At^m \leq F_{\tau}(t) \leq (1+\epsilon)At^m\quad\text{for all }t\in(0,\delta).
    \end{align*}
    Under this assumption and with $I_{a,b}$ as defined in \eqref{Iab}, we can integrate by parts and bound $I_{0,\delta}$ from above,
    \begin{align} \label{t7eq}
        I_{0,\delta} &\leq \left( 1 - \frac{\gamma(\eta,r\delta)}{\Gamma(\eta)} \right) F_{\tau}(\delta) + \frac{(1+\epsilon)Ar^{\eta}}{\Gamma(\eta)} \int_0^{\delta} t^{m+\eta-1}e^{-rt}\,\textup{d}t,
    \end{align}
    where $\gamma(s,x) := \int_0^x t^{s-1} e^{-t}\,\textup{d}t$ denotes the lower incomplete gamma function. The first term on the right-hand side of \eqref{t3eq} decays exponentially fast as $r\to\infty$ and the same is true of $I_{\delta,\infty}$, as shown in the proof of Theorem \ref{thm3}. Further, the large $r$ behavior of the integral in \eqref{t7eq} is given by
    \begin{align*}
        \int_0^{\delta} t^{m+\eta-1}e^{-rt}\,\textup{d}t \sim r^{-m-\eta} \Gamma(m+\eta)\quad\text{as }r\to\infty.
    \end{align*}
    Hence,
    \begin{align*}
        \limsup_{r\to\infty} \frac{p}{Ar^{-m}\Gamma(m+\eta)/\Gamma(\eta)} \leq 1 + \epsilon.
    \end{align*}
    An analogous argument gives the lower bound
    \begin{align*}
        \liminf_{r\to\infty} \frac{p}{Ar^{-m}\Gamma(m+\eta)/\Gamma(\eta)} \geq 1 - \epsilon.
    \end{align*}
    Repeating this procedure for $p_k$ yields similar relations,
    \begin{align*}
        \limsup_{r\to\infty} \frac{p_k}{Br^{-n}\Gamma(n+\eta)/\Gamma(\eta)} \leq 1 + \epsilon,\\
        \liminf_{r\to\infty} \frac{p_k}{Br^{-n}\Gamma(n+\eta)/\Gamma(\eta)} \geq 1 - \epsilon.
    \end{align*}
    Since $\epsilon\in(0,1)$ is arbitrary, recalling that
    \begin{align*}
        \mathbb{P}(\kappa=k | \tau<\sigma) = \frac{p_k}{p} = \frac{\int_0^{\infty} S_{\sigma}(t)\,\textup{d}F^{(k)}_{\tau}(t)}{\int_0^{\infty} S_{\sigma}(t)\,\textup{d}F_{\tau}(t)}
    \end{align*}
    completes the proof.
\end{proof}

\noindent
\begin{proof}[Proof of Theorem \ref{sharp}]
    Let $\epsilon>0$. From \eqref{Ftau1}, there exists $\delta>0$ so that
    \begin{align*}
        (1-\epsilon)\tilde{F}_{\tau}(t) \leq F_{\tau}(t) \leq (1+\epsilon)\tilde{F}_{\tau}(t)\quad\text{for all }t\in(0,\delta).
    \end{align*}
    Let $r>0$ be sufficiently large so that $1/r<\delta$ and let $I_{a,b}$ be as defined in \eqref{Iab}. Integrating by parts to bound $I_{0,1/r}$ from above,
    \begin{align*} 
        I_{0,1/r} &\leq (1+\epsilon) \int_0^{1/r} \tilde{F}_{\tau}(t) \,\textup{d}S_{\sigma}(t) = (1+\epsilon)\tilde{F}_{\tau}(1/r).
    \end{align*}
    We achieve a similar result for the lower bound:
    \begin{align*} 
        I_{0,1/r} &\geq (1-\epsilon) \int_0^{1/r} \tilde{F}_{\tau}(t) \,\textup{d}S_{\sigma}(t) = (1-\epsilon)\tilde{F}_{\tau}(1/r).
    \end{align*}
    Since $I_{1/r,\infty}$ is identically zero, we infer that
    \begin{align*}
        1 - \epsilon \leq \liminf_{r\to\infty} \frac{p}{\tilde{F}_{\tau}(1/r)} \leq \limsup_{r\to\infty} \frac{p}{\tilde{F}_{\tau}(1/r)} \leq 1 + \epsilon.
    \end{align*}
    Since $\epsilon>0$ is arbitrary, repeating this procedure for $p_k$ completes the proof.
\end{proof}

\subsection{Calculations for run-and-tumble in one dimension} \label{rrttpp}
In the right panel of Figure \ref{fig2} we illustrate conditional hitting probabilities corresponding to RTP search in one dimension between targets at $x=0$ and $x=\ell$. The searcher has probability $q\in[0,1]$ of initially moving in the negative x-direction and switches between velocities $V_1>0$ and $-V_0<0$ at Poissonian rate $\lambda>0$. We assume $x_0/V_0 < (\ell-x_0)/V_1$ and compare the asymptotic hitting probabilities for $r\gg 1$ to the hitting probabilities for any $r>0$, which we approximate numerically. Below we compute the asymptotic hitting probabilities. The details of the numerical approximation in the case that $\sigma>0$ is exponentially distributed with rate $r>0$ are in the following section.

Define $L_0:=x_0$ and $L_1:=\ell-x_0$. Integrating by parts the expression for $p:=\mathbb{P}(\tau<\sigma)$ in \eqref{pdef} yields
\begin{align*}
    p = \int_0^{\infty} f_{\sigma}(t_0+t)F_{\tau}(t_0+t)\,\textup{d}t
\end{align*}
where $t_0 = L_0/V_0$ since we assume $L_0/V_0<L_1/V_1$ and $f_{\sigma}:=\textup{d}/\textup{d}t\,F_{\sigma}$ is the density of $\sigma>0$. The cumulative distribution function of $\tau$ is given by
\begin{align*}
    F_{\tau}(t_0+t) &= \mathbb{P}(\tau\leq t_0+t)\\
    &= \mathbb{P}(\tau<t_0) + \mathbb{P}(\tau=t_0) + \mathbb{P}(t_0<\tau\leq t_0+t) + \mathbb{P}(\tau=t_0+t)\\
    &= qe^{-\lambda t_0} + \mathbb{P}(t_0<\tau\leq t_0+t).
\end{align*}
Hence,
\begin{align} \label{prtp1}
    p = qe^{-\lambda t_0} S_{\sigma}(t_0) + \int_0^{\infty} f_{\sigma}(t_0+t)\mathbb{P}(t_0<\tau\leq t_0+t)\,\textup{d}t.
\end{align}
From previous work on the short-time behavior of RTPs \cite{lawley2021pdmp}, it is known that
\begin{align} \label{prtp2}
    \mathbb{P}(t_0<\tau\leq t_0+t) \sim (1-qe^{-\lambda t_0})(\alpha_0/t_0) t\quad\text{as }t\to 0^+
\end{align}
where
\begin{align} \label{prtp3}
    \alpha_0 := \frac{\lambda L_0 (\lambda L_0q - qV_0 + V_0)}{V_0(V_0+V_1)(e^{\lambda L_0/V_0} - q)}.
\end{align}

At this point, we consider separately the resetting distributions of interest. To start, suppose $\sigma>0$ is exponentially distributed with rate $r>0$. By the proof of Theorem \ref{net_gam}, we determine the large $r$ behavior of $p>0$,
\begin{align*}
    p \sim qe^{-(\lambda+r) t_0} + e^{-rt_0}(1-qe^{-\lambda t_0})\alpha_0/(rt_0) \sim q e^{-(\lambda+r)L_0/V_0} \quad\text{as }r\to\infty.
\end{align*}
Similarly, by letting $t_1 = L_1/V_1$, we find that
\begin{align*}
    p_1 \sim (1-q) e^{-(\lambda+r)L_1/V_1}\quad\text{as }r\to\infty.
\end{align*}
Hence,
\begin{align*}
    \mathbb{P}(\kappa = 1 | \tau < \sigma) \sim \frac{1-q}{q} e^{-(\lambda+r)(L_1/V_1-L_0/V_0)}\quad\text{as }r\to\infty.
\end{align*}

Now suppose $\sigma>0$ is gamma distributed with shape $\eta>0$ and rate $r>0$. Equations \eqref{prtp1}-\eqref{prtp3} hold with
\begin{align*}
    S_{\sigma}(t) = \frac{\Gamma(\eta,rt)}{\Gamma(\eta)},\quad f_{\sigma}(t) = \frac{r^{\eta}}{\Gamma(\eta)} t^{\eta-1} e^{-rt}.
\end{align*}
From the proof of Theorem \ref{net_gam},
\begin{align*}
    p &\sim qe^{-\lambda t_0}\frac{\Gamma(\eta,rt_0)}{\Gamma(\eta)} + e^{-rt_0}(1-qe^{-\lambda t_0})\alpha_0/(rt_0)^{2-\eta}\\
    &\sim \Big(\frac{rL_0}{V_0}\Big)^{\eta-1} \frac{q}{\Gamma(\eta)} e^{-(\lambda+r)L_0/V_0} \quad\text{as }r\to\infty.
\end{align*}
Similarly, by letting $t_1 = L_1/V_1$, we find that
\begin{align*}
    p_1 \sim \Big(\frac{rL_1}{V_1}\Big)^{\eta-1} \frac{(1-q)}{\Gamma(\eta)} e^{-(\lambda+r)L_1/V_1} \quad\text{as }r\to\infty.
\end{align*}
Hence,
\begin{align*}
    \mathbb{P}(\kappa = 1 | \tau < \sigma) \sim \frac{1-q}{q}\left( \frac{L_1V_0}{L_0V_1} \right)^{\eta-1} e^{-(\lambda+r)(L_1/V_1-L_0/V_0)}\quad\text{as }r\to\infty.
\end{align*}
Setting $\eta = 1$ reduces these results to the exponential $\sigma>0$ case, as expected.

We omit the cases of uniformly distributed and sharp $\sigma>0$ since for sufficiently large $r>0$ the search process will never complete.  

\subsection{Numerical methods}
\subsubsection{Diffusion with drift in one dimension}
In Figure \ref{fig1} we illustrate conditional hitting probabilities corresponding to diffusive search in one dimension between targets at $x=0$ and $x=\ell$ with $x_0\in(0,\ell/2)$. The searcher has diffusivity $D>0$ and experiences a drift of magnitude $\mu\in\mathbb{R}$, and $\sigma>0$ is exponentially distributed with rate $r>0$. We compare the asymptotic hitting probabilities for $r\gg 1$ to the hitting probabilities for any $r>0$. The details of quadrature are as follows:

First suppose $\mu=0$. In this case, the exact conditional hitting probability for exponentially distributed $\sigma>0$ with rate $r>0$ is known \cite{aPal_2019},
\begin{align} \label{exact}
    \mathbb{P}(\kappa = 1 | \tau<\sigma) = \frac{\sinh(x_0\sqrt{r/D})}{\sinh(x_0\sqrt{r/D}) + \sinh((\ell-x_0)\sqrt{r/D})}.
\end{align}
Thus, the circular markers in Figure \ref{fig1} are given precisely by equation \eqref{exact}.

Now suppose $\mu\neq 0$ and consider the unconstrained search process. The probability density for hitting the left boundary is \cite{navarro2009}
\begin{align} \label{f0}
    f_{\tau}^{(0)}(t) := \frac{\textup{d}}{\textup{d}t} F_{\tau}^{(0)}(t) = \text{exp}\Big( -\frac{\mu x_0}{2D} -\frac{\mu^2t}{4D} \Big) \frac{D}{\ell^2} \phi\Big(\frac{Dt}{\ell^2},\frac{x_0}{\ell}\Big)
\end{align}
where
\begin{align} \label{phi}
    \phi(s,w) := \begin{cases}
        \sum_{k=1}^{\infty} e^{-k^2\pi^2s}2k\pi \sin(k\pi w),\\
        \frac{1}{\sqrt{4\pi s^3}} \sum_{k=-\infty} ^{\infty} (w+2k)e^{-(w+2k)^2/4s}.
    \end{cases}
\end{align}
The expressions for $\phi$ in \eqref{phi} are equivalent, but the top expression converges quickly for large $s$ while the bottom expression converges quickly for small $s$.

By symmetry, the probability density for hitting the right boundary is
\begin{align*}
    f_{\tau}^{(1)}(t) := \frac{\textup{d}}{\textup{d}t} F_{\tau}^{(1)}(t) = \text{exp}\Big( \frac{\mu (\ell-x_0)}{2D} -\frac{\mu^2t}{4D} \Big) \frac{D}{\ell^2} \phi\Big(\frac{Dt}{\ell^2},1-\frac{x_0}{\ell}\Big).
\end{align*}
The density for hitting either boundary is thus $f_{\tau} = f_{\tau}^{(0)} + f_{\tau}^{(1)}$ and so
\begin{align*}
    F_{\tau} = F_{\tau}^{(0)} + F_{\tau}^{(1)}.
\end{align*}

With these formulae, we use quadrature to approximate the conditional hitting probabilities from the integral representation in \eqref{pdef} and \eqref{pkdef}, which we plot with triangular and square markers in Figure \ref{fig1}. We use the short-time expansion of $\phi$ for $s\leq 1$ and the long-time expansion of $\phi$ for $s>1$. We use $10^4$ terms in each sum for $\phi$ and $10^4$ terms each in the log-spaced time intervals for $s\leq 1$ and $s>1$. We take $D=1$ and $\ell=1$. The same procedure is used to produce Figure \ref{fig3}.

\subsubsection{Diffusion between concentric spheres in three dimensions}
Similar to the case of diffusion in one dimension, here we use quadrature to approximate the conditional hitting probabilities from the integral representation in \eqref{pdef} and \eqref{pkdef} where $\sigma>0$ is exponentially distributed with rate $r>0$. We determine $F_{\tau}$ and $F_{\tau}^{(1)}$ by numerically solving the partial differential equations that they satisfy. 

In particular, the cumulative distribution function of Brownian motion between concentric spheres of radii $R_0>0$ and $R_1>0$ with $R_0<R_1$ and initial position $x_0\in(R_0,R_1)$ satisfies
\begin{align} \label{conc}
    \frac{\partial F_{\tau}}{\partial t} = D\left( \frac{2}{x_0} \frac{\partial F_{\tau}}{\partial x_0} + \frac{\partial^2 F_{\tau}}{\partial x_0^2} \right),\quad x_0\in(R_0,R_1)
\end{align}
with initial condition $F_{\tau} = 0$ at $t=0$ and boundary conditions $F_{\tau} = 1$ at $x_0=R_0$ and $x_0=R_1$. The function $F_{\tau}^{(1)}$ similarly satisfies \eqref{conc} with the same initial condition and boundary conditions $F_{\tau} = 0$ at $x_0=R_0$ and $x_0=R_1$. We numerically solve these boundary value problems using the MATLAB solver \verb+pdepe+ with 2000 linearly spaced radial mesh points between $R_0$ and $R_1$ and 5000 log-spaced time points between $10^{-10}$ and $10$ \cite{MATLAB}. In the left panel of Figure \ref{fig2} we take $D=1$, $R_0=1$, and $R_1=2$.

\subsubsection{Run-and-tumble in one dimension}
In the right panel of Figure \ref{fig2} we illustrate conditional hitting probabilities corresponding to an RTP in one dimension with parameter choices as described in subsection \ref{RTP}. The circular markers in the figure correspond to simulations of the process wherein $q=1/2$, $\ell=1$, $\lambda=2$, and $V_0=V_1=2$. Each markers corresponds to the average of 8000 trials.

\subsubsection{Superdiffusive L\'evy flight in one dimension}
In Figure \ref{levy} we illustrate conditional hitting probabilities corresponding to a superdiffusive L\'evy flight in one dimension as detailed in subsection \ref{superdiff}. We now describe the stochastic simulation algorithm used herein. Given a discrete time step $\Delta t >0$, we generate a statistically exact path of the $(\alpha/2)$-subordinator $U = \{U(t)\}_{t\geq 0}$ on the discrete time grid $\{t_n\}_{n\in \mathbb{N}}$ with $t_n = n\Delta t$ via
\begin{align*}
    U(t_{n+1}) = U(t_n) + (\Delta t)^{2/\alpha} \Theta_n,\quad n\geq 0,
\end{align*}
where $U(t_0) = U(0) = 0$ and $\{\Theta_n\}_{n\in\mathbb{N}}$ is an independent and identically-distributed sequence of realizations
\begin{align*}
    \Theta = \frac{\sin((\alpha/2)(V + \pi/2))}{(\cos V)^{2/\alpha}} \left( \frac{\cos (V - (\alpha/2)(V + \pi/2))}{E} \right)^{(2 - \alpha)/\alpha},
\end{align*}
where $V$ is uniformly distributed on $(-\pi/2,\pi/2)$ and $E$ is an independent unit mean exponential random variable \cite{carnaffan2017}. From this we can generate a statistically exact path of Brownian motion $\{B(s)\}_{s\geq 0}$ on the discrete time grid $\{U(t_n)\}_{n\in\mathbb{N}}$ via
\begin{align*}
    B(U(t_{n+1})) = B(U(t_n)) + \sqrt{2(D_s\Delta t)^{2/\alpha}\Theta_n}\xi_n,\quad n\geq 0,
\end{align*}
where $\{\xi_n\}_{n\in\mathbb{N}}$ is an independent and identically-distributed sequence of standard Gaussian random variables. Finally one obtains a statistically exact path of a L\'evy flight with generalized diffusivity $D_s>0$ via the following random time change of the Brownian motion $B$,
\begin{align*}
    X(t) := D_s^{1/\alpha} B(U(t)),\quad t\geq 0.
\end{align*}

To determine the conditional hitting probabilities, we generate an exponentially distributed resetting time and simulate a L\'evy flight path until the reset time. If the path reaches a target prior to the resetting time, the process ends and we note which target was found. Otherwise we reset the search to its initial position and start the process anew. Repeating this process for a given resetting rate allows us to compute the corresponding conditional hitting probability. Each data point in Figure \ref{levy} is computed from 12000 independent trials with $D_s = 1$ and $\Delta t = 10^{-4}$.

\section*{Acknowledgements}
This material is based upon work supported by the National Science Foundation Graduate Research Fellowship Program under Grant No. 2139322. SDL was supported by the National Science Foundation (Grant Nos. DMS-2325258 and DMS-1944574). Any opinions, findings, and conclusions or recommendations expressed in this material are those of the authors and do not necessarily reflect the views of the National Science Foundation.


\newpage
\bibliography{library.bib}
\bibliographystyle{unsrt}

\end{document}